\documentclass{article}
\usepackage{ijcai18}

\usepackage{times}
\usepackage{xcolor}
\usepackage{soul}
\usepackage[utf8]{inputenc}
\usepackage[small]{caption}

\usepackage{amsmath}
\usepackage{amsfonts}
\usepackage{amsthm}
\usepackage{mathrsfs}
\usepackage{graphicx}
\usepackage{booktabs}
\usepackage{framed}

\title{Customer Sharing in Economic Networks with Costs}

\author{
	Bin Li$^1$,
	Dong Hao$^1$,
	Dengji Zhao$^2$\and Tao Zhou$^1$
	\\
	$^1$University of Electronic Science and Technology of China, Chengdu, China\\
	$^2$ ShanghaiTech University, Shanghai, China
	\\libin@std.uestc.edu.cn, haodong@uestc.edu.cn, zhaodj@shanghaitech.edu.cn, zhutou@ustc.edu
}

\begin{document}

\maketitle

\begin{abstract}  
	In an economic market, sellers, infomediaries and customers constitute an economic network. Each seller has her own customer group and the seller's private customers are unobservable to other sellers. Therefore, a seller can only sell commodities among her own customers unless other sellers or infomediaries share her sale information to their customer groups. However, a seller is not incentivized to share others' sale information by default, which leads to inefficient resource allocation and limited revenue for the sale. To tackle this problem, we develop a novel mechanism called customer sharing mechanism (CSM) which incentivizes all sellers to share each other's sale information to their private customer groups. Furthermore, CSM also incentivizes all customers to truthfully participate in the sale. In the end, CSM not only allocates the commodities efficiently but also optimizes the seller's revenue.
\end{abstract}
\section{Introduction}
In many markets such as online shopping platforms, there are multiple sellers and many potential customers. Each seller has her own customer group and different seller's customer groups may not overlap. The interactive relationships among these sellers and customer groups form an economic network. For example, in online search engines or social media platforms (e.g., Google, Twitter, WordPress), it is common to see various ads from some commodity sellers. These social media platforms usually have their own customers which could be the registered users or those who often browse these platforms. By advertising in these platforms, a seller could enlarge her selling quantity and at the same time, the platforms also increase its profit from these paid advertisements.
Supply chain system is an another example, every dealer in the system serves a local customer group that is geographically close to her and usually if a dealer is not a terminal retail trader, she may have secondary dealers. Based on commercial contracts, the dealers in the system form an economic network where each edge represents the relationship between supply and marketing.
Other similar economic interactions also emerge in logistics, routing and job-hunting networks.
When one seller (or advertiser, dealer) wants to sell her own commodity for a good price, without inviting other sellers to become mediators for her and diffuse the sale information in the economic network, usually she can only sell among customers that she can directly inform, although it would be very likely that a high price bidder hides in another seller's customer group. As a consequence, this
results in a locally optimal allocation and the seller's revenue cannot be globally maximized. Such a sub-optimal allocation scenario is very common in many real-world markets. Even in non-competitive markets where the sellers are not rivals to each other, a seller's local sale information is hard to be shared by other sellers.


The underlying reason for such inefficiency is mainly due to the fact that sharing others' sale information could be costly.
The information sharing action itself is usually costless, for example, a sharer can do it by posting the sales information via twitter or facebook. However, a more important cost could potentially exist. In the economic networks, if a high bid buyer is not a member of the seller's direct customer group, to find this buyer and complete a sale with a globally optimized revenue, a trading chain from the seller to the winning buyer via one or multiple information sharers must be established.
Every seller or infomediary along the chain could have a cost, for example the costs for shipping the commodity. For intermediate sellers who are selling the same or similar commodities, this cost could be their potential customer loss; for those noncompetitive intermediate sellers or infomediaries, the costs could simply be monetary commission demands, the information handling fee or transportation/labor costs for shipping the commodity.
It is important that such costs should be covered by a mechanism if a seller wishes to diffuse her sale information among the economic network. However, as far as we know, it is of particular difficulty for the existing market mechanisms to incentivize sellers with costs to share others' sale information. As a result, the highest bidder in the network cannot be discovered and the globally optimal allocation cannot be achieved.

In this paper, we investigate the above customer sharing problem under an economic network setting in the view of mechanism design. 
Taking into consideration the transaction costs, we propose two mechanisms which can be used to incentivize the sellers to become mediators and diffuse other sellers' sale information.
The first mechanism is an extension of the recently proposed information diffusion mechanism (IDM) \cite{li2017mechanism}. We show that if the economic network forms a tree structure,
the extended mechanism is individually rational (IR), incentive compatible (IC), budget balanced (BB) and efficient. Nevertheless, the extended mechanism fails to work in general graphs. Therefore, we further develop a novel mechanism called customer sharing mechanism (CSM). We prove that CSM is IR, IC, BB and efficient in general cases and the revenue generated by CSM is always higher than that given by the Vickrey auction.

There exists a rich body of work studying mechanism design problems that relates to networks. Nevertheless, the structure of involved networks are usually assumed to be fixed and prior known to the designer and problems raised are also essentially different from ours, for instance the cost-sharing problem for multicast transmissions considered in \cite{Moulin2001Strategyproof,feigenbaum2001sharing}, the frugal path problem investigated in \cite{archer2007frugal,elkind2004frugality} and products procurements studied in supply networks \cite{babaioff2009mechanisms,chen2005efficient}.
In terms of promoting information spreading in networks, previous works can be divided into two categories in general. For non-strategic agents settings, agents act according to the predefined physics rules and the aim is to investigate the scale of diffusion under various rules and target influential nodes that can trigger large diffusion cascades in social networks \cite{kempe2003maximizing,rogers2010diffusion,pastor2015epidemic}. In settings with strategic agents, the most widely studied scenario is to incentivize agents to invite more people to accomplish a challenge together$-$for a few examples, see \cite{kleinberg2005query,leskovec2007dynamics,pickard2011time,emek2011mechanisms,cebrian2012finding}. However, none of them considered an auction setting and amount of money is needed to compensate the sharers.
A most related work is \cite{li2017mechanism} where authors proposed an auction mechanism which not only incentivizes agents to propagate the sale information to their neighbours in a social network, but also improves the allocation efficiency and seller's revenue comparing with the Vickrey auction.

The structure of the paper is organized as follows. Section 2 describes the model of the customer sharing problem. Section 3 gives a solution in tree-structure economic networks based on \cite{li2017mechanism}. To tackle the problem in general scenarios, we design a novel mechanism called customer sharing mechanism (CSM) and analyze its properties in Section 4. Finally, we conclude in Section 5.

\section{The Model}
\newtheorem{defn}{Definition}
\newtheorem{prop}{Proposition}
\newtheorem{lemma}{Lemma}
\newtheorem{theorem}{Theorem}
\newtheorem{corollary}{Corollary}
\newtheorem{exmp}{Example}


Consider a \textit{seller} $s$ selling a single commodity in an economic network. Besides the seller, the economic network includes $n$ agents denoted by $N=\{1,2,\cdots,n\}$, which are divided into two categories: \textit{buyers} (customers) who are interested in buying this commodity and \textit{intermediate nodes} who are other sellers or infomediaries. Each intermediate node can only directly link to her own buyer group and some other intermediate nodes. The buyer groups of intermediate nodes may not overlap. We assume the buyers do not communicate with each other, i.e., there are no links between the buyers.

Each buyer $i\in N$ assigns a private value $b_i\geq 0$ to the commodity $-$ the maximum amount she is willing to pay for the commodity. Each intermediate node $i$ possesses a private neighbour set $r_i\subseteq N$ with whom $i$ can directly exchange the sale information. Furthermore, an intermediate node incurs a cost $c_i$ when a trade passes through her hands. We assume cost $c_i$ is known to the seller once intermediate node $i$ participates in the auction. 
Initially, only the seller's neighbours $r_s$ are informed of the sale.

Figure \ref{trading_network} depicts two examples of the economic network. The value inside each circle is the buyer's private valuation for the commodity. The value (except the seller $s$) inside each square is the intermediate node's cost.
For instance, intermediate node $B$ suffers a loss of $10$ if the winner is buyer $F$ since a trade from the seller to buyer $F$ must pass through $B$'s hands. 



\begin{figure}[t]
	\centering
	\includegraphics[width=3.3in]{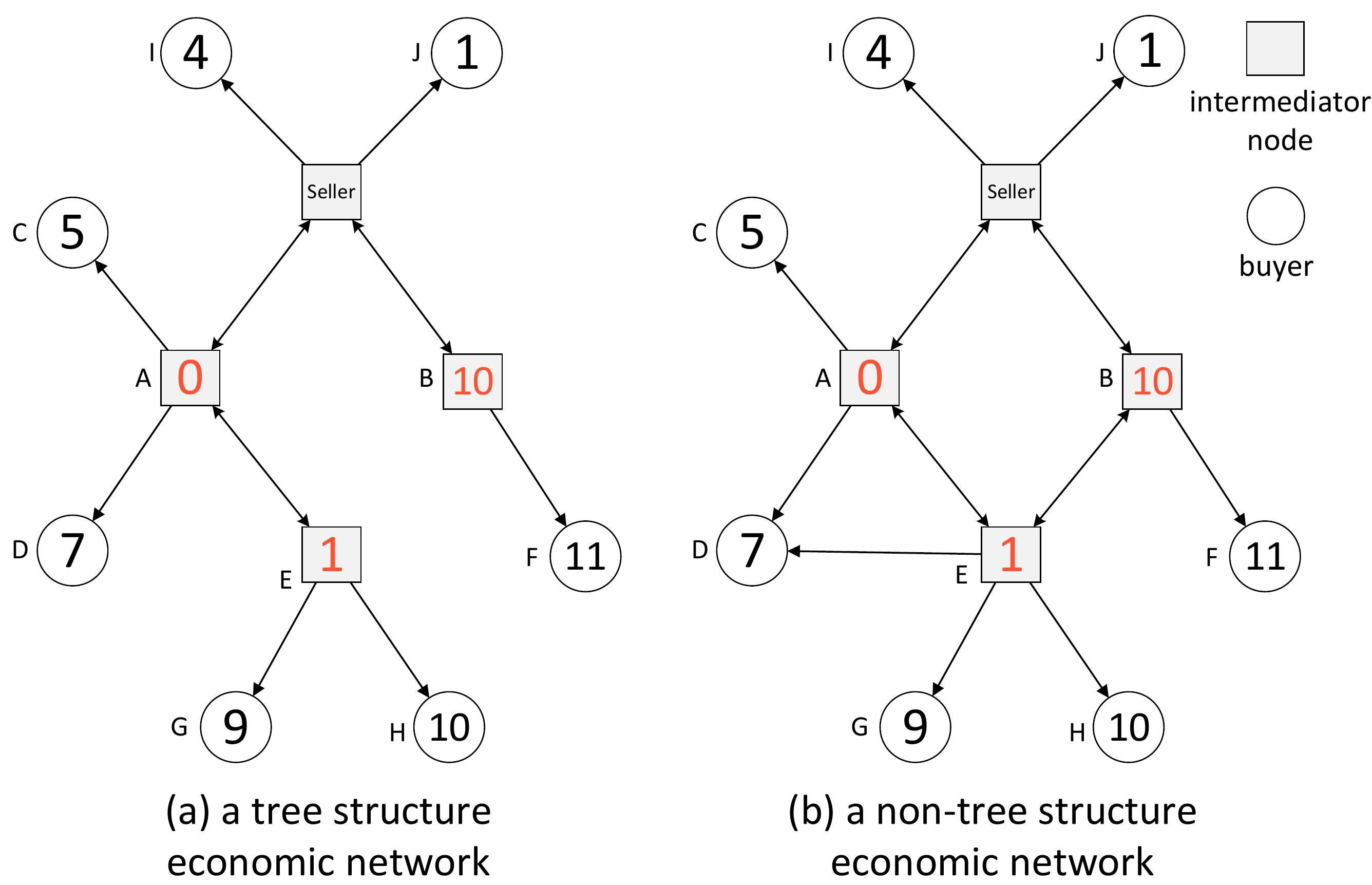}
	\caption{Two instances of the economic network.}\label{trading_network}
\end{figure}

Formally, let $t_i$ be the \textit{private type} of agent $i$. In our setting, we have 
\begin{equation*}
t_i=
\begin{cases}
b_i &\text{if $i$ is a buyer,}\\
r_i &\text{if $i$ is an intermediate node.}\\
\end{cases}
\end{equation*}
Denote $t=(t_1,\cdots,t_n)$ as the type profile of all agents and $t_{-i}$ be the type profile of all agents except $i$. We have $t=(t_i,t_{-i})$. Let $T_i$ be the type space of agent $i$ and $T$ be the type profile space of all agents. Specifically, $T_i=\mathcal{P}(N)$ for any intermediate node $i$ where $\mathcal{P}(N)$ is the power set of $N$, $T_i=\mathbb{R}_{\geq 0}$ for any buyer $i$ and $T=\times T_{i\in N}$.

The mechanism asks each agent, who is aware of the sale, to act according to her true type. Denote $t_i'=b_i'\in T_i$ as the reported type of buyer $i$ and $t'_i=r_i'\in T_i$ as intermediate node $i$'s reported type which means $i$ shares the sale information to her neighbours in $r_i'$. The reported type of intermediate node $i$ is limited to $\mathcal{P}({r_i})\subseteq T_i$ as she is not aware of others who are not her neighbours. Let $t'=(t_1',\cdots,t_n')$ where $t'_i=nil$ if $i$ has never been informed of the auction or $i$ does not want to participate.

\begin{defn}
	Given a type report profile $t'$ of all agents, a \textit{trading chain} from seller $s$ to agent $i$ is a sequence of agents $(a_1,\cdots,a_k,a_{k+1},\cdots,a_p,i)$ such that $a_1\in r_s$ and for $1<l\leq p$, $a_l\in r_{a_{l-1}}',i\in r_{a_p}'$ and no agent appears twice in the sequence, i.e., it is a simple path from seller $s$ to agent $i$.
\end{defn}
A trading chain demonstrates how a commodity routes in the economic network from $s$ to $i$ through each intermediary $a_k$ successively. To completing a trade in an economic network, the commodity has to go through one or more intermediate nodes if the winner is not a member of the seller's direct neighbours. Correspondingly, the \textit{transaction costs} of the trade is the total costs incurred by these intermediate nodes. 
Since an agent who is not informed of the sale cannot participate in the auction, not all type report profiles are feasible.

\begin{defn}
	We say a type report profile $t'$ is feasible if for all $i\in N$, $t_i'\neq nil$ iff there exists a trading chain from seller $s$ to $i$ following the type reports of $t_{-i}^\prime$. 
\end{defn}

A feasible type report profile reflects a practical information sharing process where seller $s$ is the initial sharing agent. In what follows, we only design auction mechanisms under feasible type report profiles. 
\begin{defn}
	Given a feasible type report profile $t'$, for agent $i$ with $t_i'\neq nil$, denote $LCC_i$ as the trading chain with the lowest transaction costs from seller $s$ to $i$, i.e., $LCC_i={\arg\min}_{l_i\in \mathcal{L}_i}{\sum_{j\in l_i\setminus\{i\}}{c_i}}$ where $\mathcal{L}_i$ is the set of all trading chains from seller $s$ to $i$.
\end{defn}
For buyer $i$, $LCC_i$ represents the most economic path when trading a commodity from the seller to $i$ through intermediate nodes.
\begin{defn}
	An \textit{auction mechanism} $\mathcal{M}=(\pi,x)$ has the following components: an \textit{allocation policy} $\pi = \{\pi_i\}_{i\in N}$ and a \textit{payment policy} $x=\{x_i\}_{i\in N}$, where $\pi_i:T\rightarrow \{-1,0,1\}$ and $x_i:T\rightarrow \mathbb{R}$ are the allocation and payment policy for $i$ respectively.
\end{defn}

Given a feasible type report profile $t'$, $\pi_i(t') = 1$ means $i$ wins the commodity, otherwise $\pi_i(t') = -1$ if $i$ is assigned to the trading chain to the winner and $\pi_i(t') = 0$ if $i$ is neither the winner nor an agent in the trading chain. $x_i(t') \geq 0$ indicates that $i$ pays the auctioneer $x_i(t')$, and $i$ receives $|x_i(t')|$ from the auctioneer if $x_i(t') < 0$. We say $\pi$ is a \textit{feasible allocation} if for all $t' \in T$, the mechanism only chooses at most one buyer as the winner and the agents with $\pi_i(t') \neq 0 $ form a trading chain from seller $s$ to the winner. 
Hereinafter we only consider feasible allocations. 

Let $v_i(t_i,\pi(t'))$ be the valuation function of agent $i$ under feasible allocation $\pi(t')$ which denotes her value about the outcome $\pi(t')$ without considering her payment. In our setting, 
\begin{equation*}
v_i(t_i,\pi(t'))=
\begin{cases}
b_i &\text{$\pi_i(t')=1$,}\\
-c_i &\text{$\pi_i(t')=-1$,}\\
0 &\text{$\pi_i(t')=0$.}
\end{cases}
\end{equation*}
Notice that an intermediate node with $\pi_i(t')=-1$ means that she is a member of the selected trading chain to the winner. Since a cost $c_i$ is incurred when trading the commodity, her value about the allocation is $-c_i$ without considering her payment.

Given a feasible allocation $\pi(t')$, the \textit{social welfare} of the allocation is defined as $W(\pi(t'))=\sum_{i\in N}v_i(t'_i,\pi(t'))$. Because only agents in the selected trading chain have non-zero valuations, we have $W(\pi(t'))=\sum_{i\in TC_w}v_i(t'_i,\pi(t'))$ where $w$ is the winner and $TC_w$ is the selected trading chain to $w$.
\begin{defn}
	An allocation $\pi$ is \textit{efficient} if for any feasible type report profile $t'$,
	$\pi \in {\arg\max}_{\pi' \in \Pi} W(\pi'(t'))$
	where $\Pi$ is the set of all feasible allocations. Denote $W^*(\pi(t'))$ as the social welfare under efficient allocation $\pi$.
\end{defn}
Equivalently, an auction mechanism $(\pi,x)$ with efficient allocation allocates the commodity to buyer $m={\arg\max}_{i\in N}\sum_{j\in LCC_i}v_j(t'_j,\pi(t'))$ and selects $LCC_m$ as the trading chain. In the following, denote the expression $\sum_{j\in LCC_i}v_j(t'_j,\pi(t'))$ by $SW_i$.

Given agent $i$ of truthful type $t_i$, a feasible type report profile $t'$ and an auction mechanism $(\pi, x)$, the \textit{utility} of $i$ under the allocation $\pi(t')$ and the payment $x(t')$ is quasilinear and is defined as: $u_i(t_i, t', (\pi,x)) = v_i(t_i,\pi(t')) - x_i(t')$.

A mechanism is individually rational if for each buyer, her utility is non-negative when she reports her valuation truthfully, and for each intermediate node, her utility is non-negative no matter to whom she shares the sale information and what the others do, i.e., she is not forced to share the sale information to her neighbours.

\begin{defn}
	A mechanism $(\pi,x)$ is \textit{individually rational} (IR) if $u_i(t_i, (b_i, t'_{-i}), (\pi,x)) \geq 0$ for all buyer $i$ and $u_i(t_i, (r'_i, t''_{-i}), (\pi,x)) \geq 0$ for all intermediate node $i$, where $(r'_i, t''_{-i})$ is the feasible type report profile when $i$ reports $r_i'\in \mathcal{P}({r_i})$.
\end{defn}

Incentive compatibility in our setting means that for all buyers and intermediate nodes, reporting valuation truthfully and spreading the auction information to all neighbours in the economic network is a dominant strategy.
\begin{defn}
	A mechanism $(\pi, x)$ is \textit{incentive compatible} (IC) if
	$u_i(t_i, (t_i,t_{-i}^\prime), (\pi,x)) \geq u_i(t_i, (t_i^\prime, t_{-i}^{\prime\prime}), (\pi,x))$ for all $i\in N$, where $(t_i^\prime, t_{-i}^{\prime\prime})$ is the feasible type report profile when $i$ reports $t_i'$.
\end{defn}
Given a feasible type report profile $t^\prime$ and a mechanism $\mathcal{M} = (\pi, x)$, the seller's \textit{revenue} generated by $\mathcal{M}$ is defined by the sum of all agents' payments, denoted by $Rev^{\mathcal{M}}(t^\prime) = \sum_{i\in N} x_i(t^\prime)$.

\begin{defn}
	A mechanism $(\pi, x)$ is \textit{weakly budget balanced} if for all feasible type report profile $t'\in T$, $Rev^{\mathcal{M}}(t') \geq 0$.
\end{defn}

In rest of this paper, we design mechanisms that satisfy above properties. In order to economize on notations, let $v_i(t'_i,\pi(t'))$, $W^*(\pi(t'))$ be $v_i(t'_i,\pi)$, $W^*(t')$ respectively.

\section{Information Diffusion Mechanism with Transaction Costs}
The popular VCG mechanism \cite{vickrey1961counterspeculation,clarke1971multipart,groves1973incentives} can directly apply in our setting. In fact, it is IR, IC and efficient. One drawback of the VCG mechanism is that it is not weakly budget balanced and the deficit it suffers could be linear in $|N|$$-$the number of agents. The proof is trivial and we omit it here. In this section, we first derive ideas from information diffusion mechanism (IDM) proposed in \cite{li2017mechanism} and extends it to economic networks with costs.
In following we show that the extended mechanism performs well when the underlying economic network forms a tree structure but fails to be truthful in general scenarios. 

We say agent $i$ is agent $j$'s \textit{diffusion critical agent} if all the trading chains started from seller $s$ to $j$ have to pass $i$. Intuitively, if $i$ is $j$'s diffusion critical agent under type report profile $t'$, $j$ will not be able to receive the sale information if $i$ has not joined the auction. Denote $d_i$ as the set of all agents who share $i$ as their diffusion critical agent. Let $-d_i=N\setminus d_i$ and $t'_{-d_i}$ denote the type report profile of agents in $-d_i$. In Figure $1(a)$, agent $E$ is the shared diffusion critical agent of agent $H$, $G$ and herself, therefore we have $d_E=\{E,G,H\}$. 

\begin{defn}\label{Diffusion_Critical_Sequence}
	Given the agents' type report profile $t^\prime$, for each $i\in N$, define $C_i=\{s_1,s_2,\cdots,s_k,i\}$ as the \textit{diffusion critical sequence} of $i$, which is an ordered set of all $i$'s diffusion critical agents and the order is determined by the relation $d_{s_1} \supset d_{s_2} \supset ,\cdots, \supset d_{s_k}\supset d_{i} $.
\end{defn}


An agent in the sequence can only receive the sale information if all agents ordered above the agent have received it. It is worth noting that agent $i$'s diffusion critical sequence is the intersection set of all trading chains from the seller to $i$. Now we give our mechanism based on the notion of diffusion critical sequence.


\begin{framed}
	\noindent\textbf{Information Diffusion Mechanism with Transaction Costs (IDM-TC)}\\
	\rule{\textwidth}{0.5pt}
	\begin{itemize}
		\item Given a feasible type report profile $t'$, allocate the commodity to buyer $m={\arg\max}_{i\in N}SW_i$ and choose $LCC_m$ as the trading chain (i.e., efficient allocation with random tie-breaking). Denote the ordered set $\{1,2,\cdots,m-1,m\}$ by $C_m$.
		\item The payment policy for each agent is given as:
		\begin{small}
		\begin{equation*}\label{midm_payment}
		\begin{cases}
		W^*_{-d_i}  - W^*_{-d_{i+1}} - c_i & \text{if $i \in C_m\setminus\{m\}$,}\\
		-c_i & \text{if $i \in LCC_m\setminus C_m$,}\\
		W^*_{-m} +  \sum_{j\in LCC_m\setminus \{m\}} c_j & \text{if $i=m$,} \\
		0 & \text{otherwise,} 
		\end{cases}
		\end{equation*} 
		\end{small}
	where $W^*_{-x}$ denotes $W^*(t'_{-x})$ for short.
	\end{itemize}
\end{framed}

IDM-TC generalizes the ideas of IDM such that only the diffusion critical agents of the winner may have a positive utility. Other non-diffusion-critical agents who are assigned to the trading chain from seller $s$ to the winner only get a payoff to cover their costs. To ensure the revenue of the seller, the transaction costs of the trade are covered by the winner.

We give a running example to show how IDM-TC works in Figure $1(a)$. According to the allocation policy of IDM-TC, buyer $H$ wins the commodity since $H={\arg \max}_{i\in N} SW_i$. Buyer $H$ can't join the auction without the participation of $A$ or $E$, therefore $H$'s diffusion critical agents are $\{A,E,H\}$ and $C_H=\{A,E,H\}$ which is the same as $LCC_H$. According to the payment policy, winner $H$ pays $W^*_{-H}+(c_A+c_E)=8+(0+1)=9$. Intermediate node $E$ pays $W^*_{-d_E}-W^*_{-d_H}-c_E=-2$. In a similar way, $A$ pays $-3$ and all other agents pay zero. Finally, the revenue of seller $s$ is $9+(-2)+(-3)=4$. One can verify that all agents in Figure $1(a)$ will act truthfully, that is all buyers will report their true valuations and all intermediate nodes will share the sale information to all their neighbours once they are aware of the auction. In the following, we show that IDM-TC performs well when the underlying economic network forms a tree structure.

\begin{prop}\label{didm_tree}
	IDM-TC is efficient, individually rational, incentive compatible and weakly budget balanced if the economic network forms a tree structure. The revenue of the seller given by IDM-TC is at least the revenue given by the Vickrey auction.
\end{prop}
\begin{proof}
	It is obvious that IDM-TC is efficient according to the allocation policy. Next, we give a sketch for proving IC and IR properties. 
	
	For winner $m$, a higher bid does not change her utility ($\geq 0$) since she still wins. If she bids a lower value such that she loses the commodity, her utility becomes zero. For buyer $i$ who is not the winner, it is no good for her to become a winner by bidding higher, otherwise she will get a negative utility since $SW_i \le SW_m\subseteq -d_i$. Therefore it is a buyer's best strategy to bid truthfully. For any intermediate node $i\notin C_m$, no matter to whom she shares the sale information to, $LCC_m$ is not changed since $LCC_m=C_m\subseteq -d_i$ when the economic network forms a tree structure. This means $i$ is still not in $LCC_m$ and her utility is always zero. If $i\in C_m$, since $W_{-d_i}^*$ is independent of her, $i$'s utility is maximized by spreading the auction information to all her neighbours in which case $W_{-d_{i+1}}$ is maximized. Therefore fully spreading the sale information is the best reply for intermediates nodes. 
	
	The revenue of seller is $\sum_{i\in N} x_i(t')=W^*_{-d_1}\geq 0$. Since $r_s\subseteq -d_1\cup\{1\}$, we have $W^*_{-d_1}\geq b_{r_s}^2$ where $b_{r_s}^2$ is the second highest bid in $r_s$ (i.e., the revenue given by the Vickrey auction).
\end{proof}

Although IDM-TC is powerful when the economic network forms a tree structure, it fails to be truthful in general cases. For example, the winner is still buyer $H$ if there is an edge between intermediate nodes $B$ and $E$ in Figure 1(a). However, intermediate node $A$ is no longer a diffusion critical agent of $H$ and she can only obtain a payment of $c_A$ which covers her cost only. To achieve a higher revenue, intermediate node $A$ can refuse to share the sale information to intermediate node $E$. In this case, the winner becomes $D$ and intermediate node $A$ becomes $D$'s diffusion critical agent. Eventually, she will get a payoff of $1$ except her cost according to the payment policy.


The main reason IDM-TC fails to be truthful is because non-diffusion-critical agents of the winner have influence on allocation of the mechanism, which does not happen if there are no transaction costs. The utility change induces more complicate propagation strategies for intermediate nodes and are not fully embodied in IDM-TC. To compensate intermediate nodes that incur costs and incentivize them to share the sale information to all their neighbours in general economic networks, we will develop new mechanism in the next section.

\section{Customer Sharing Mechanism}
In this section, we design a novel efficient mechanism that is IR, IC and weakly budget balanced in general economic networks. This mechanism, called customer sharing mechanism, captures the contributions of every agent in an economic network, 
which is different from IDM-TC that only focuses on the diffusion critical agents of the winner. Before introducing our mechanism, we first define an edge-related set for each intermediate node which plays a vital role in our mechanism.

\begin{defn}\label{minimum_set}
	 Given the agents' type report profile $t'$, for each intermediate node $i\in N$, define $i$'s \textit{threshold neighbourhood} ${{r_i^*}'}$ as the minimum subset of $r_i'$ that makes the winner under efficient allocation changed if $i$ does not share the sale information to ${r_i^*}'$, i.e., ${r_i^*}'=\arg \min_{r_i''\subseteq r_i'} \{|r_i''|\}$ where $\pi_{m'}(r_i'\setminus r_i'', t_{-i}'')=1 \land \pi_m(r_i', t_{-i}')=1 \land m'\neq m$ and $\pi$ is an efficient allocation.
\end{defn} 

Take intermediate node $A$ in Figure $1(b)$ for an example, $r_A^*=\{E\}\subset r_A=\{C,D,E\}$ is the threshold neighbourhood of $A$ since the winner changes to buyer $D$ if $A$ does not share the sale information to intermediate node $E$. If $i$ cannot change the winner through any diffusion strategies, then set ${{r_i^*}'} = r_i'$. Next, we show that the threshold neighbourhood of an intermediate node is well defined. For any intermediate node $i$ that does not belong to $LCC_m$, it is clear that ${r_i^*}'=r_i'$ since $LCC_m$ cannot be affected by $i$'s diffusion strategies. For $i\in LCC_m$, we can construct her threshold neighbourhood in the following way:
\begin{itemize}
	\item[1.] initialize ${r_i^*}'=\{\emptyset\}$ and let $m$ be the winner in efficient allocation.
	\item[2.] add $i$'s neighbour $j\in LCC_m$ to ${r_i^*}'$ and remove the edge $(i,j)$ from the network and move to step 3.
	\item[3.] if $m'=m$ and $i\in LCC_{m}$ where $m'$ is the winner in the new network, then back to step 2; if $m'=m$ and $i\notin LCC_m$, then return $i$'s threshold neighbourhood  as $r_i'$ and terminate this procedure; otherwise, finish the procedure and return ${r_i^*}'$ as $i$'s threshold neighbourhood.
\end{itemize}

Clearly, for any $j\in {r_i^*}'$, the order of deleting edge $(i,j)$ does not affect the outcome of above procedure since the lowest cost trading chain from $k(\neq j) \in {r_i^*}'$ to $m$ is not affected. Furthermore, the procedure adds only one $i$'s threshold neighbour in step 2. Therefore the outcome is unique and the output threshold neighbourhood is minimized.
Based on the notion of threshold neighbourhood, we give our mechanism in what follows.

\begin{framed}
	\noindent\textbf{Customer Sharing Mechanism (CSM)}\\
	\rule{\textwidth}{0.5pt}
	\begin{itemize}
		\item \textbf{Allocation policy:} Given a feasible type report profile $t'$, allocate the commodity to buyer $m={\arg\max}_{j\in N}SW_j$ and trade the commodity along $LCC_m$ (with random tie-breaking).
		\item \textbf{Payment policy:} The payment policy is defined for each category of agents as follows.
		\begin{itemize}
		\item for buyer $i\in N$, her payment is defined as:
		\begin{equation*}
			W^*_{-i}  - W^*(t') + v_i(t'_i,\pi^{csm}).
		\end{equation*}
		\item for an intermediate node $i$, her payment is:
		\begin{equation*}
			W^*_{{-d_{i}}}  - W^*_{-{r_i^*}'} + v_i(t'_i,\pi^{csm}),
		\end{equation*}
		where $W^*_{-{r_i^*}'}$ denotes the maximum social welfare under feasible type report profile $(r_i'\setminus{r_i^*}',t_{-i}'')$.
		\end{itemize}

	\end{itemize}
\end{framed}
Intuitively, in the customer sharing mechanism a buyer pays the VCG payment which equals to the social welfare decrease of the others caused by her participation. Each intermediate node's payment is correlated with her threshold neighbourhood. The peculiar structure in an intermediate node's payment motivates her to share the sale information to all her neighbours. At the same time, seller's revenue is also guaranteed. Next, we prove that CSM satisfies all properties we expect.

\begin{theorem}
	CSM is efficient, individually rational and incentive compatible.
\end{theorem}
\begin{proof}
	CSM is efficient according to the allocation policy. Subsequently, we prove that CSM is individually rational and incentive compatible by two steps as follows.
	
	\item 1. For any buyer $i$, if she reports $b_i$ truthfully, i.e., $b_i' = b_i$, then $u_i(t_i,t',(\pi,x))=W(t')^*-W_{-i}^*$. Since $W(t')^*$ is the optimal social welfare under the constraints that $i$ cannot influence, if $i$ can misreport $b_i'$ to change the allocation to increase $v_i(t_i,\pi^{csm}) + (W(t')^*-v_i(t'_i,\pi^{csm}))$, then it contradicts that $W(t')^*$ is the optimal social welfare. Furthermore, $W_{-i}^*$ is independent of $i$ and we have $W(t')^*\geq W_{-i}^*$. Therefore, $i$'s utility is maximised when she reports $b_i'$ truthfully.
	
	
	\item 2. For any intermediate node $i$, we show that sharing the sale information to all her neighbours maximizes her utility. If intermediate node $i$ doesn't belong to $LCC_m$, then we have $v_i(t_i,\pi^{csm})=v_i(t'_i,\pi^{csm})=0$. Note that $LCC_m\subseteq -d_{i}\subseteq -{r^*_i}'$, then $W^*_{{-d_{i}}} = W^*_{-{r^*_i}'}=SW_m=W^*(t')$. Thus her utility is zero and cannot be affected by her sharing strategies since $i\notin -d_i$. For intermediate node $i\in LCC_m$, her utility is $W^*_{-{r^*_i}'}-W^*_{{-d_{i}}}$ by reason of $v_i(t_i,\pi^{csm})=v_i(t'_i,\pi^{csm})=-c_i$, which is no less than zero because of $-d_i\subseteq{-{r^*_i}'}$. Firstly, the latter term $W^*_{{-d_{i}}}$ is independent of $r_i'$ since $i\notin -d_i$ and we only consider how the former term is affected by different sharing strategies of intermediate node $i$. When intermediate node $i$ misreports her type as $r_i'\subseteq r_i$, there are two possible outcomes. If the winner has changed to another buyer, then we must have $r_i^*\subset r_i\setminus r_i'$ since $r_i^*$ is the minimum subset of $r_i$ that makes the current winner changed. Notice that the threshold neighbourhood ${r^*_i}'$ is the set of agents to whom intermediate node $i$ does not share the sale information when she reports $r_i'\subseteq r_i$, therefore we get $-{r_i^*}'\subset -r_i^*$. If the winner is not changed, then we have $r_i^*\subseteq ({r_i^*}'\cup (r_i\setminus r_i'))$ which also means that $-{r_i^*}'\subset -r_i^*$. In either case, we have $-{r_i^*}'\subseteq -r_i^*$ and $W^*_{-{r^*_i}'}\leq W^*_{-r^*_i}$. Thus sharing the sale information to all neighbours maximizes an intermediate node's utility. 
	
	Therefore, CSM is individually rational and incentive compatible, that is biding truthfully and spreading the sale information to all their neighbours is a dominant strategy for all agents in the economic network.
\end{proof}

The reason of truthfulness of an intermediate node is due to the fact that the more neighbours an intermediate node shares the sale information to, the smaller her threshold neighbourhood will become. Therefore, her utility is maximized when she shares the sale information to all her neighbours because $W^*_{-{r^*_i}'}$ is maximized.

It is worth noting that the identification of threshold neighbourhood is a key component in the proof of incentive compatibility. Furthermore, the threshold neighbourhood also plays a crucial role in the budget balanced property of CSM.


\begin{lemma}\label{change_winner}
	Given a feasible type report profile $t'$, assume buyer $m$ wins the commodity and $LCC_m=\{l_1,l_2,\cdots,l_q=m\}$. For any intermediate node $l_i \in LCC_m$, if the winner becomes $m'\neq m$ when $l_i$'s sharing strategy changes from $r_{l_i}'$ to $r_{l_i}'\setminus{r_{l_i}^*}'$, then $\{l_{i+1},l_{i+2},\cdots,l_q\}\cap LCC_{m'}=\emptyset$.
\end{lemma}
\begin{proof}
	Assume there is an intermediate node $l_j\in\{l_{i+1},l_{i+2},\cdots,l_q\}\cap LCC_{m'}$. Denote the lowest transaction costs from seller $s$ to $l_j$ as $x$, $l_j$ to buyer $m$ as $y$ and $l_j$ to buyer $m'$ as $z$ under type report profile $(r_{l_i}'\setminus{r_{l_i}^*}',t_{-{l_i}}'')$. Since $m'$ is the new winner, we have $b_{m'}-z-x>b_m-y-x$ which means $b_{m'}-z>b_m-y$ which contradicts the fact that $LCC_m$ is the trading chain with lowest transaction costs from seller $s$ to buyer $m$ under type report profile $t'$.
\end{proof}
It is straightforward that intermediate node $i$'s utility is zero if her diffusion strategies cannot change the current winner according to the payment policy of CSM, in which case ${r_i^*}'=r_i'$. And for the intermediate nodes who can change current winner to another one, $\{l_{i+1},l_{i+2},\cdots,l_q\}\cap LCC_{m'}=\emptyset$ means that $LCC_{m'}$ is independent of the agent $l_k\in \{l_{i+1},l_{i+2},\cdots,l_q\}$. Therefore we have $LCC_{m'}\subseteq -d_{l_k}$ which leads to $W^*_{-{r_i^*}'}=SW_{m'}= \sum_{i\in LCC_{m'}} v_i(t_i',\pi^{csm})\leq \sum_{i\in -d_{l_k}} v_i(t_i',\pi^{csm})=W^*_{{-d_{l_k}}}$. Based on this observation, we show CSM is budget balanced.

\begin{theorem}
	CSM is weakly budget balanced and the revenue of the seller is no less than $W^*_{-d_{1^*}}$ where $1^*$ is the first agent in $LCC_m$ with $r_{1^*}^*\neq r_{1^*}$.
\end{theorem}
\begin{proof}
	Denote the ordered sequence $LCC_m^*=\{1^*,2^*,\cdots,k^*\}\subseteq LCC_m$ as the set of intermediate nodes with $r_{i}^*\neq r_i$. For an intermediate node that doesn't belong to $LCC_m^*$, her payment is $v_i(t_i,\pi^{csm})$ because of $W_{{-d_{i}}}^{*}=W_{-r_i}^{*}=W_{-r_i^*}^{*}$. Therefore the revenue of seller $Rev^{csm}$ is given by:
	\begin{equation*}
	\begin{aligned}
	&\sum_{i\in N}{x_i(t)}=\sum_{i\in LCC_m}{x_i(t)}\\
	&=\sum_{i\in LCC_m^*}(W_{{-d_{i}}}^{*}  - W_{-r_i^*}^{*}) + W^{*}_{-m}  - W^*({t})\\
	&+ \sum_{i\in LCC_m}v_i(t_i,\pi^{csm})\\
	\end{aligned}
	\end{equation*}
	
	\begin{equation*}
	\begin{aligned}
	&=\sum_{i\in LCC_m^*}(W_{{-d_{i}}}^{*}  - W_{-r_i^*}^{*})+W^{*}_{-m}\\
	&=W_{{-d_{1^*}}}^{*}+\sum_{i\in LCC_m^*\setminus \{1^*\}}(- W_{-r_i^*}^{*}+W_{{-d_{i+1}}}^{*}),\\
	\end{aligned}
	\end{equation*}
	where $-d_{k^*+1}$ is the set of $-m$.
	According to Lemma \ref{change_winner}, we have that for any intermediate node $i^*\in LCC_m^*$, ${(i+1)}^*\notin LCC_{m'}$ where $m'$ is the winner under feasible type report profile $(r_{i^*}\setminus r_{i^*}^*,t'_{-i^*})$. This means that $LCC_{m'}\subseteq -d_{(i+1)^*}$. Therefore $W_{{-d_{j+1}}}^{*}\geq W_{-r_j^*}^{*}$ for all $j\in LCC_m^*\setminus\{1^*\}$, which leads to $\sum_{i\in LCC_m^*\setminus\{1^*\}}(- W_{-r_i^*}^{*}+W_{{-d_{i+1}}}^{*})\geq 0$. That is, $Rev^{csm}=W_{{-d_{1^*}}}^{*}+\delta \geq W_{{-d_{1^*}}}^{*} \geq b_{r_s}^2\geq 0$ where $\delta$ is a non-negative value.
\end{proof}
We run CSM in Figure $1(b)$ to end this section. By the allocation policy, identify the winner in Figure $1(b)$ as buyer $H={\arg\max}_{j\in N}SW_j$ with $LCC_H=\{A,E,H\}$. Then determine the threshold neighbourhood for each intermediate node in $LCC_H$: $r_A^*=\{E\}$ and $r_E^*=\{H\}$. According to the payment policy of CSM, $x_A=W^*_{{-d_{A}}}  - W^*_{-{r_A^*}} + v_i(t_i,\pi^{csm})=4-7+(-0)=-3$. In a similar way, we have $x_E=7-8+(-1)=-2$ and $x_H=8-9+10=9$. All other agents who don't belong to $LCC_H$ pay zero. Therefore, the revenue of the seller is $-3+(-2)+9=4$ which is four times of the revenue given by the Vickrey auction.
\section{Conclusions}

To increase seller's revenue and improve allocation efficiency in economic networks with costs, we design a novel auction mechanism called customer sharing mechanism. In this mechanism, all the sellers and infomediaries are incentivized to diffuse the auction information to all their neighbours and all the customers are incentivized to bid truthfully. As a consequence, all the sellers can share customers to each other and the commodity of each seller will be allocated globally, both the revenue and efficiency are significantly improved comparing with the Vickrey auction in which agents have no incentives to share information. 

One premise in the model was that the costs of intermediate nodes are assumed to be known. In fact, public costs widely exist in real-world. In e-markets, online platforms usually charge a fixed and public-known commission when dealing with a trade. In freight systems, the costs are consumed resources for shipping the commodity. When selling privacy datasets via a computer network, the costs are electricity consumption when handling the datasets. These aforementioned costs can be easily estimated and verified in practice. Additionally, when the costs are also private information for intermediate nodes, one cannot design mechanisms that satisfy all the good properties simultaneously. Consider a simple network: one seller, one intermediate node with a private cost and one customer with private valuation constitute a line. Since the intermediate node can affect the allocation of the commodity by changing her diffusion strategies, it is clear that we can reduce this setting to a bilateral trade. According to Myerson–Satterthwaite theorem \cite{Myerson1983Efficient}, no truthful, IR and efficient mechanism is budget balanced. 

There are several directions for future researches. In this paper, we give one solution for customer sharing problem, however, to maximize the seller's revenue, we need to characterize the conditions for all truthful mechanisms in the economic network. Moreover, the seller in our paper has only one commodity for sale, it is of great interest to design mechanisms for selling multiple commodities. Without considering costs, a generalized IDM is proposed for selling $k$ homogeneous items via social networks where each buyer wants at most one item \cite{Zhao2018Multi}. The proposed mechanism can also be further implemented by using techniques from Distributed Algorithmic Mechanism Design (DAMD) \cite{Feigenbaum2002DistributedAM,feigenbaum2005bgp}.

\section*{Acknowledgments}
This work was partially supported by the National Natural
Science Foundation of China (NNSFC) under Grant No.
71601029.
\bibliographystyle{named}
\bibliography{ijcai18csm}

\end{document}